\documentclass{article}

\usepackage[a4paper, left=30mm, right=20mm, top=20mm, bottom=20mm]{geometry}

\usepackage{graphicx}

\usepackage[utf8]{inputenc}
\usepackage[T2A,T3,T1]{fontenc}

\DeclareSymbolFont{tipa}{T3}{cmr}{m}{n}
\DeclareMathAccent{\invbreve}{\mathalpha}{tipa}{16}

\usepackage{amsmath, amsfonts, amssymb, bm}

\usepackage{yfonts}
\usepackage{mathrsfs}
\usepackage{pifont}
\usepackage{ifsym}
\usepackage{bbm}

\DeclareMathAlphabet\mathbfcal{OMS}{cmsy}{b}{n}
\usepackage{MnSymbol}
\usepackage{xcolor}
\usepackage[colorlinks=true]{hyperref}
\usepackage{caption}
\usepackage{subcaption}
\usepackage{tabularx}
\usepackage{epigraph} 
\usepackage{algorithm}
\usepackage{todonotes}

\definecolor{darkblue}{rgb}{0.0, 0.0, 0.55}
\definecolor{forestgreen}{rgb}{0.13, 0.55, 0.13}
\definecolor{darkteal}{rgb}{0.0, 0.24, 0.18}

\hypersetup{
  linkcolor=darkblue,
  urlcolor=darkteal,
  citecolor=forestgreen
}

\definecolor{w}{RGB}{231, 160, 76}
\definecolor{l}{RGB}{37, 84, 138}

\usepackage{tikz-cd}
\usetikzlibrary{decorations.pathreplacing,fit,shapes.geometric}

\usepackage{savesym}

\usepackage{environ}

\usepackage{amsthm}
\usepackage{etoolbox} 

\theoremstyle{definition}

\newtheorem{definition}{Definition}[section]
\AtBeginEnvironment{definition}{%
  \pushQED{\qed}%
}
\AtEndEnvironment{definition}{\popQED}

\newtheorem{theorem}{Theorem}[section]
\AtBeginEnvironment{theorem}{%
  \pushQED{\qed}%
}
\AtEndEnvironment{theorem}{\popQED}

\newtheorem{lemma}{Lemma}[section]
\AtBeginEnvironment{lemma}{%
  \pushQED{\qed}%
}
\AtEndEnvironment{lemma}{\popQED}

\AtBeginEnvironment{conjecture}{%
  \pushQED{\qed}%
}
\AtEndEnvironment{conjecture}{\popQED}

\newtheorem*{remark}{Remark}
\AtBeginEnvironment{remark}{%
  \pushQED{\qed}%
}
\AtEndEnvironment{remark}{\popQED}

\newtheorem*{example}{Example}
\AtBeginEnvironment{example}{%
  \pushQED{\qed}%
}
\AtEndEnvironment{example}{\popQED}

\AtEndEnvironment{lemma}{\popQED}

\DeclareMathOperator\supp{supp}

\usepackage{rotating}

\title{Non-Cooperative Games with Uncertainty \\
        \large Definition, Existence and some properties \\ of the Extended Equilibrium }
\author{József Konczer  \\
        \href{mailto:konczer.j@gmail.com}{konczer.j@gmail.com},
        \href{https://konczer.github.io/}{konczer.github.io}
        }
\date{February 2025}

\begin{document}

\maketitle

\begin{abstract}
    This paper introduces a framework for finite non-cooperative games where each player faces a globally uncertain parameter with no common prior. Every player chooses both a mixed strategy and projects an emergent subjective prior to the uncertain parameters. We define an ``Extended Equilibrium'' by requiring that no player can improve her expected utility via a unilateral change of strategy, and the emergent subjective priors are such that they maximize the expected regret of the players. A fixed-point argument -- based on Brouwer's fixed point theorem and mimicking the construction of Nash -- ensures existence. Additionally, the ``No Fictional Faith'' theorem shows that any subjective equilibrium prior must stay non-concentrated if the parameter truly matters to a player. 
    This approach provides a framework that unifies regret-based statistical decision theory and game theory, yielding a tool for handling strategic decision-making in the presence of deeply uncertain parameters.
\end{abstract}

\section{Introduction}

In this paper, multiplayer non-cooperative games with uncertain parameters will be introduced. A so-called Extended Equilibrium concept will be defined, and the existence of such equilibria will be proved. Further, a general property of the equilibrium solution will be proved, ensuring that no unjustified certainty is emerging about the uncertain factors.

The construction of multiplayer games with uncertainty can be viewed as the generalization of the theoretical framework explored in \cite{arxiv:StatisticalGames}\footnote{For further resources see the project's \href{https://github.com/Konczer/UncertaintyTheory/tree/main/StatisticalGames}{GitHub repository}.}, which framed a single Agent's decision-making in the face of uncertainty as a game with a fictional player. 
This game theoretical framework of ``statistical decision-theory'' \cite{book:Wald,book:Savage,book:CoxStatistics} defines equilibrium strategies -- in imaginary games played with adversarial fictional players controlling the unknown parameters -- and interprets these strategies as decision-making heuristics for the Agent.

To embed decision-making problems into a game theoretical framework, it is crucial to appropriately choose the imagined utility function of the fictional player. There are multiple arguments supporting the choice of regret as the fictional player's utility function. (For historical context, see the following references \cite{book:Savage,paper:MinimaxRegretNiehans1948,book:CoxStatistics}; for further philosophical arguments, see \cite{essay:EssayOnUncertaintyKonczer2024}). 
In the context of Bayesian (or Harsányi-) games, a similar construction has been proposed in \cite{arxiv:Hyafil2012Regret}.
(In general, our construction can be viewed as a specific prior selection requirement for games with incomplete information with inconsistent beliefs \cite{book:GameTheory}. In this interpretation, the extended equilibrium is simply a Bayesian equilibrium with a set of emergent subjective priors \cite{book:GameTheory}, while the priors are determined by a minimax-regret principle.)

The existence proof of the extended equilibrium is mainly based on an elegant construction used by John Nash in his paper from 1951 \cite{paper:Nash1951}, directly utilizing the Brouwer fixed-point theorem.
The ``No Fictional Faith'' theorem is based on basic properties of the regret matrices, and the result can be viewed as a weakened regularity property of the solution. This condition is often a prescribed requirement for the prior in the context of Bayesian statistics, while in extended equilibrium -- its weaker form -- emerges naturally.

\section{An illustrating fictional story}

Although this section is optional and is not essential to understand the formal results,
it aims to motivate the framework and ground the equilibrium solution with a fictional story using a relatable semi-historical narrative:

\begin{example}[The story of two generals and the weather]
\label{story:GW}
Imagine two generals, a ``Defender'' and an ``Attacker''\footnote{To embed the story to a semi-historical context, one can picture the dilemmas of Scipio and Hannibal in the context of the Second Punic War in 3rd century BC, or alternatively, one can depict generals and emperors from the Tang-Tibetan Wars in the 7th century. In the first case, the Alps could provide the mountain scene, while in the second case, the Himalayas -- or more realistically, other mountains of Central Asia and the Tibetan Plateau -- could play this role.}. The ``Defender'' and the ``Attacker'' are preparing for battle 
near a mountainous region more than a thousand years ago -- in an era when technology for reliable weather forecasting was not accessible.
The possible battlegrounds are new for both generals, but scouts identified two directions for the attack -- an upper route in the hills and a lower route in a valley.
In the hills, the ever-changing weather also plays a crucial role. By a sudden storm, it can wipe out a whole attacking army.

Both generals can choose from two actions. The ``Defender'' can decide to defend the upper or the lower route, while the ``Attacker'' can decide to attack from the same two directions. The next day, the weather can be in two states; there can be a storm in the hills, or the weather can stay calm.

The winning and losing prospects depend on the actions of the generals and on the weather in the following way:
If the weather is calm, then the ``Defender'' loses if he is attacked from an undefended direction, but the defence is always successful if both generals choose the same direction.
In case of a storm, however, the attacking army would be lost if it dared the upper route.
(On the lower route, in the valley, the weather does not change the outcome of a battle.)

Assuming that no general has reliable knowledge about the weather or can reliably associate a probability distribution to it -- and their lack of information is common knowledge -- what strategy should the generals follow?

(The generals are capable of randomization -- methods such as divinations of various kinds were used from prehistory to make randomized actions \cite{paper:DivinationArticle}. A mixed strategy, therefore, will be interpreted as a probability distribution representing chances by which specific actions can be expected from the ``Defender'' and the ``Attacker''.)

\end{example}

The following constructions and the introduced equilibrium concepts aim to grasp strategic decision-making problems, such as the dilemma depicted in the previous story.

\section{Definitions and notation}

\subsection{Games with uncertainty}

In this section, finite non-cooperative games with globally uncertain factors will be formally defined.

\begin{definition}[Games with uncertainty]
\label{def:GameWithUncertainty}
A (finite, $N$-person, non-cooperative) game with (global) uncertainty is a tuple $\langle N, \underline{\mathcal{A}}, \Theta, \underline{U} \rangle$, where:

\begin{itemize}
    \item $N \in \mathbb{N}$ is the number of players;
    \item $\underline{\mathcal{A}} = (\mathcal{A}_1,\dots,\mathcal{A}_N)$ is the action set profile of the players, where all $\mathcal{A}_i$ are finite sets;
    \item $\Theta$ is a finite set of the globally unknown parameter or state. (Global uncertainty means that no player has credible side-information about it's value, and this is common knowledge);
    \item $\underline{U} = (U_1,\dots,U_N)$ is the set of subjective, action and parameter-dependent utility functions taking real values: $U_i : \underline{\mathcal{A}} \times \Theta \mapsto \mathbb{R}$.
\end{itemize}
    
\end{definition}

A few remarks can be made comparing the concept with previous classical game constructions appearing in the literature.

\begin{remark}
    Superficially, this might appear as an game with $N+1$ players, where the extra $N+1$-th player is associated with $\Theta$ as her action set. However, the above-defined game with uncertainty does not associate utilities with the uncertain parameter, i.e. only the $N$ real players have specified utilities. This crucial difference prevents us from simply understanding games with uncertainty as multiplayer games with one extra player.

\end{remark}

\begin{remark}
    Another remarkable feature of this construction is that it does not require a commonly known probability distribution (or prior) $\pi$ on the unknown parameters. Therefore, the construction cannot be understood simply as a Bayesian (or Harsányi) game \cite{book:GameTheory} either.
    
\end{remark}

\begin{remark}
    In this construction we do not specify a set of subjective prior believes either, therefore it can not be interpreted simply as a game with incomplete information with inconsistent beliefs \cite{book:GameTheory} -- in which a set of subjective priors $\underline{\pi}$ should be specified for the uncertain set of parameters.
    
\end{remark}

\subsubsection{Concrete example}
\label{sec:ExampleDef}

To give a concrete example, let us see how the story of the ``Generals and the weather'' -- presented in \autoref{story:GW}  -- could be formally described as a game with uncertainty:

\begin{example}[Generals and the weather]
\label{ex:DefGeneralsAndTheWeather}
The ``Generals and the weather'' game, $\mathcal{GW}$ can be formalized as a two-player game with uncertainty:
$\mathcal{GW} = \langle N = 2 , \underline{\mathcal{A}}, \Theta, \underline{U} \rangle$:

\begin{itemize}
    \item The two players are:
    \begin{itemize}
        \item Player 1 being the ``Defender'';
        \item Player 2 being the ``Attacker''.
    \end{itemize}
    \item $\underline{\mathcal{A}} = (\mathcal{A}_1,\mathcal{A}_2)$, where:
    \begin{itemize}
        \item $\mathcal{A}_1 = \{\text{``Up''},\text{``Down''}\}$ the possible actions of Player 1, the ``Defender'': choosing to go up to defend the upper route or go down to defend the lower route;
        \item $\mathcal{A}_2 = \{\text{``Up''},\text{``Down''}\}$ the possible actions of Player 2, the ``Attacker'': choosing to go up to attack the upper route or go down to attack the lower route.
    \end{itemize}
    \item $\Theta = \{\text{``Calm''}, \text{``Storm''}\}$ is the possible states of the globally unknown weather;
    \item $\underline{U} = (U_1, U_2)$ is the set of subjective, action and parameter-dependent utility functions taking real values: $U_i : \underline{\mathcal{A}} \times \Theta \mapsto \mathbb{R}$.
    (Here, we make the choice that winning and losing will have $1$ and $0$ utility respectively, i.e. the utility represents the chance of winning in a given scenario.)
    \begin{itemize}
        \item The utility function of Player 1, the ``Defender'' (captured in a matrix form where rows and columns are representing the first and second player's actions respectively), looks the following: 
        \[
        U_1(.,.;\text{``Calm''}) =
        \begin{bmatrix}
            1 & 0 \\
            0 & 1
        \end{bmatrix}, \quad
        U_1(.,.;\text{``Storm''}) =
        \begin{bmatrix}
            1 & 0 \\
            1 & 1
        \end{bmatrix}
        \]
        \item Using the same convention, the  utility function of Player 2, the ``Attacker'', has the following form: 
        \[
        U_2(.,.;\text{``Calm''}) =
        \begin{bmatrix}
            0 & 1 \\
            1 & 0
        \end{bmatrix}, \quad
        U_2(.,.;\text{``Storm''}) =
        \begin{bmatrix}
            0 & 1 \\
            0 & 0
        \end{bmatrix}
        \]
    \end{itemize}
\end{itemize}

\end{example}

\subsection{Notation}

To define further concepts and the extended equilibrium, it is useful to introduce and clarify some notation: 

\begin{definition}[Value assignment / substitution]
The notation:
\begin{equation}
    \chi^{\leftarrow \theta} := \theta
\end{equation}
means that a variable $\chi$ is replaced by $\theta$. 

For an $N$-tuple $\underline{b} = (b_1,\dots,b_N)$ and an index $i \in \{1,\dots,N\}$, the notation
    \begin{equation}
        \underline{b}^{[i] \leftarrow a_i} 
        :=
        (b_1,\dots,b_i^{\leftarrow a_i},\dots,b_N)
        =
        (b_1,\dots,a_i,\dots,b_N)
    \end{equation}
    
denotes the tuple obtained by replacing only the $i$th entry $b_i$ of $\underline{b}$ with $a_i$, while leaving all other entries unchanged\footnote{Naturally, for $i=1$: $\underline{b}^{[1] \leftarrow a_1} := (a_1,\dots,b_N)$ and for $i=N$: $\underline{b}^{[N] \leftarrow a_N} := (b_1,\dots,a_N)$.}.

\end{definition}

\begin{definition}[Mixed strategies and subjective priors]
We will denote the $i$th player's (potentially mixed) strategy by $\sigma_i$, which is a probability distribution on the finite set $\mathcal{A}_i$ and the $i$th player's fictional or subjective prior by $\pi_i$ which is a probability distribution on the finite set $\Theta$. Formally:

    \begin{equation}
        \sigma_i \in \Delta(\mathcal{A}_i), \quad \pi_i \in \Delta(\Theta)
    \end{equation}
\end{definition}

Where for a finite set $S$ we denote the set of probability distributions on it by $\Delta(S)$:

\begin{equation}
    \Delta(S) := \left \{p: S \mapsto [0,1] \ \bigg | \ \sum_{s \in S} p(s) = 1 \right \}
\end{equation}

\begin{definition}[Complete strategy profile]
A complete strategy profile $\underline{\rho}$ contains both the real strategies of the $N$ real players, $\underline{\sigma} = (\sigma_1,\dots,\sigma_N) \in \mathcal{S}$ and the fictional/imaginary strategies or subjective priors of the players, projected to the uncertain parameters $\underline{\pi}=(\pi_1,\dots,\pi_N) \in \mathcal{P}$:

    \begin{equation}
        \underline{\rho} = ( \underline{\sigma}; \underline{\pi} )
    \end{equation}

The set of complete strategy profiles $\mathcal{R}$ can be constructed as the Descartes product of probability spaces:
    
    \begin{equation}
        \underline{\rho} \in 
        \mathcal{R} :=
        \mathcal{S} \times \mathcal{P} =
        \left ( \bigtimes_{i=1}^N \Delta(\mathcal{A}_i) \right ) \times
        \left ( \Delta(\Theta) \right )^N
    \end{equation}

\end{definition}

\begin{definition}[Product probabilities]
To shorten the definitions, the following notation for the -- partially subjective -- probability of an action profile and a parameter value for the $i$th player, given a complete strategy profile, can be introduced: 

    \begin{equation}
        \Pi_i(\underline{a};\theta | \underline{\rho}) =
        \Pi_i(\underline{a};\theta | \underline{\sigma};\underline{\pi}) := 
        \left ( \Pi_{j=1}^N \sigma_j(a_j) \right ) \ \pi_i(\theta)
    \end{equation}

And the probability of an action profile given a -- real -- strategy profile:

    \begin{equation}
        \Pi(\underline{a}|\underline{\sigma}) := \Pi_{j=1}^N \sigma_j(a_j)
    \end{equation}

\end{definition}

\subsection{Regret and Expected quantities}

A central concept in the definition of the extended equilibrium for games with uncertainty is the personal regret of a player. This quantity depends on the player, her action, the value of the unknown parameter and the ``background'' that other players create by their probabilistic strategies.
Formally, we define the quantity in the following way:

\begin{definition}[Personal regret]

\begin{equation}
    R_i(\underline{a};\theta | \underline{\sigma}) = 
    \max_{c_i \in \mathcal{A}_i} 
    \left ( 
    \sum_{\underline{b} \in \underline{\mathcal{A}}} U_i(\underline{b}^{[i] \leftarrow c_i};\theta) \ \Pi(\underline{b}|\underline{\sigma})
    \right )
    -U_i(\underline{a};\theta)
\end{equation}

\end{definition}

\begin{remark}[Harsh regret]
There would be another way to introduce a similar concept. We will call it the harsh version of regret, where we do not take the maximum of the average but the average of the maximum, i.e. 
$R^H_i(\underline{a};\theta | \underline{\sigma}) = 
    \sum_{\underline{b}}
    \left ( 
    \max_{c_i} 
     U_i(\underline{b}^{[i] \leftarrow c_i};\theta) 
    \right )
    \ \Pi(\underline{b}|\underline{\sigma})
    -U_i(\underline{a};\theta)$.
However, a framework for the uncertain parameters built on this harsher definition of regret fails to fulfil crucial consistency requirements in the context of nested decision-making dilemmas. (For details, see construction in \cite{essay:EssayOnUncertaintyKonczer2024}.)

\end{remark}

\begin{definition}[Expected Utilities and Expected Regrets]
\label{def:EUER}

We can define the following -- partially subjective -- expected utilities and regrets, assuming a complete strategy profile $\underline{\rho} = ( \underline{\sigma}; \underline{\pi} )$:

    \begin{equation}
        \mathrm{EU_i}(\underline{\sigma};\underline{\pi}) = 
        \sum_{\underline{a} \in \underline{\mathcal{A}};\theta \in \Theta}
        U_i(\underline{a};\theta) \ 
        \Pi_i(\underline{a};\theta|\underline{\sigma};\underline{\pi})
    \end{equation}

    \begin{equation}
        \mathrm{ER_i}(\underline{\sigma};\underline{\pi}) = 
         \sum_{\underline{a} \in \underline{\mathcal{A}};\theta \in \Theta}
        R_i(\underline{a};\theta|\underline{\sigma}) \ 
        \Pi_i(\underline{a};\theta|\underline{\sigma};\underline{\pi})
    \end{equation}

Similarly we can introduce -- partially subjective -- expected utilities, given that the $i$th player is making action $a_i \in \mathcal{A}_i$  (assuming the same complete strategy profile):

    \begin{equation}
        \mathrm{EU_i}(a_i|\underline{\sigma};\underline{\pi}) = 
        \sum_{\underline{b} \in \underline{\mathcal{A}};\theta \in \Theta}
        U_i(\underline{b}^{[i] \leftarrow a_i};\theta) \ 
        \Pi_i(\underline{b};\theta|\underline{\sigma};\underline{\pi})
    \end{equation}

and the expected regrets, given that the unknown parameter is set to $\theta \in \Theta$:

    \begin{equation}
    \label{defeq:ER_from_Regret}
        \mathrm{ER}_i(\theta | \underline{\sigma};\underline{\pi}) = 
        \sum_{\underline{a} \in \underline{\mathcal{A}};\chi \in \Theta}
        R_i(\underline{a};\chi^{\leftarrow \theta} |\underline{\sigma}) \ 
        \Pi_i(\underline{a};\chi|\underline{\sigma};\underline{\pi})
    \end{equation}

Which -- because $\sum_\chi \pi_i(\chi) = 1$ -- simplifies to:

    \begin{equation}
        \mathrm{ER}_i(\theta | \underline{\sigma};\underline{\pi}) = 
        \sum_{\underline{a} \in \underline{\mathcal{A}}}
        R_i(\underline{a};\theta|\underline{\sigma}) \ 
        \Pi(\underline{a}|\underline{\sigma})
    \end{equation}
    
\end{definition}

\begin{remark}
\label{altdef:EUER}
    For an alternative definition of the expected regret, one can introduce action and state-dependent expected or effective utility matrices for each player:

    \begin{equation}
    \label{defeq:EUMatrix}
        \mathrm{EU}_i(a_i;\theta | \underline{\sigma}) = 
        \sum_{\underline{b} \in \underline{\mathcal{A}}} U_i(\underline{b}^{[i] \leftarrow a_i};\theta) \ \Pi(\underline{b}|\underline{\sigma})
    \end{equation}

    ($\sum_{b_i} \sigma_i(b_i) = 1$, therefore this expected utility matrix depends only on the strategies of other, $j \ne i$ players. (Conventionally denoted as $\sigma_{-i}$ \cite{book:GameTheory}.))

    From this expected or effective utility matrix, we get the expected utility for a certain action by taking the weighted sum with respect to the subjective prior $\pi_i$:
    
    \begin{equation}
        \mathrm{EU}_i(a_i | \underline{\sigma};\underline{\pi}) = 
        \sum_{\theta \in \Theta} \mathrm{EU}_i(a_i;\theta | \underline{\sigma}) \ \pi_i(\theta)
    \end{equation}

    From the expected or effective utility matrix, we can introduce an effective regret matrix:

    \begin{equation}
    \label{eq:RegretMatrix}
        \mathrm{ER}_i(a_i;\theta | \underline{\sigma}) =
        \max_{c_i \in \mathcal{A}_i} \left ( \mathrm{EU}_i(c_i;\theta | \underline{\sigma}) \right ) - \mathrm{EU}_i(a_i;\theta | \underline{\sigma})
    \end{equation}

    One can define the expected regret of the $i$th player for a given value of the unknown parameter $\theta$ as:

    \begin{equation}
    \label{eq:ERthetaERMatrix}
        \mathrm{ER}_i(\theta | \underline{\sigma}) = 
        \sum_{a_i \in \mathcal{A}_i} \mathrm{ER}_i(a_i;\theta | \underline{\sigma}) \ \sigma_i(a_i)
    \end{equation}

    Finally, the expected regret for the $i$th player is:

    \begin{equation}
        \mathrm{ER}_i(\underline{\sigma};\underline{\pi}) =
        \sum_{\theta \in \Theta} \mathrm{ER}_i(\theta | \underline{\sigma}) \ \pi_i(\theta)
    \end{equation}
    
\end{remark}

For later use, we can spell out four simple lemmas about the expected quantities and the regret matrix.

\begin{lemma}
    The matching expected quantities defined in Def. \autoref{def:EUER} and in Remark \ref{altdef:EUER} are equivalent.
    
\end{lemma}

\begin{lemma}[Non-negativity of the effective regret matrix]
\label{lemma:NonNegativeRegret}
    The effective regret matrix defined in \eqref{eq:RegretMatrix} is always non-negative:
    \begin{equation}
        \forall i, \ \forall \underline{\sigma} \in \mathcal{S}, \ \forall a_i \in \mathcal{A}_i, \ \forall \theta \in \Theta \quad
        \mathrm{ER}_i(a_i;\theta | \underline{\sigma}) \ge 0
    \end{equation}
    
\end{lemma}

\begin{lemma}[Existence of a regret-less action]
\label{lemma:ExistenceBestResponse}
    Further, because all action sets are finite, for any parameter value $\theta \in \Theta$, there is at least one action $a^o$ which yields zero effective regret (as defined in \eqref{eq:RegretMatrix}):

    \begin{equation}
        \forall i, \ \forall \underline{\sigma} \in \mathcal{S}, \ \forall \theta \in \Theta, \ \exists a^o \in \mathcal{A}_i \quad
        \mathrm{ER}_i(a^o;\theta | \underline{\sigma}) = 0
    \end{equation}
    
\end{lemma}

\begin{lemma}
\label{lemma:EUER}
     Following from the previous definitions and the concept of the support of a probability distribution, we get:
     
    \begin{equation}
        \mathrm{EU_i}(\underline{\sigma};\underline{\pi})
        =
        \sum_{a_i \in \mathcal{A}_i} \mathrm{EU_i}(a_i|\underline{\sigma};\underline{\pi}) \ \sigma_i(a_i)
        =
        \sum_{a_i \in \supp(\sigma_i)} \mathrm{EU_i}(a_i|\underline{\sigma};\underline{\pi}) \ \sigma_i(a_i)
    \end{equation}

    \begin{equation}
        \mathrm{ER_i}(\underline{\sigma};\underline{\pi})
        =
        \sum_{\theta} \mathrm{ER}_i(\theta | \underline{\sigma};\underline{\pi}) \ \pi_i(\theta)
        =
        \sum_{\theta \in \supp(\pi_i)} \mathrm{ER}_i(\theta | \underline{\sigma};\underline{\pi}) \ \pi_i(\theta)
    \end{equation}
    
\end{lemma}

\subsection{Extended Equilibrium}

Similarly to the definition of the Nash equilibrium \cite{paper:Nash1951,book:GameTheory}, we can introduce the concept of an extended equilibrium, which stands for a complete strategy profile -- containing both a real strategy profile and a set of fictional/imaginary strategies or subjective priors for all players -- which fulfils the following inequalities:

\begin{definition}[Extended Equilibrium]
\label{def:ExEq}
We call a complete strategy profile $\underline{\rho}^* = (\underline{\sigma}^*;\underline{\pi}^*) \in \mathcal{R}$ an extended equilibrium if it fulfils the following set of inequalities:

    \begin{equation}
    \label{defeq:ExtendedEquilibriumEU}
        \forall i, \forall a_i \in \mathcal{A}_i, \  \mathrm{EU}_i(\underline{\sigma}^*;\underline{\pi}^*) \ge \mathrm{EU}_i(a_i | \underline{\sigma}^*;\underline{\pi}^*)
    \end{equation}
    \begin{equation}
    \label{defeq:ExtendedEquilibriumER}
        \forall i, \forall \theta \in \Theta, \  \mathrm{ER}_i(\underline{\sigma}^*;\underline{\pi}^*) \ge \mathrm{ER}_i(\theta | \underline{\sigma}^*;\underline{\pi}^*)
    \end{equation}

\end{definition}

In words, the extended equilibrium is a set of strategies and beliefs, such as no player expects increased expected utility by unilaterally deviating from her strategy, and all players project emergent subjective priors to the uncertain parameters such that it maximizes their expected regret.

\subsubsection{Example of an extended equilibrium}

For a concrete example, consider the ``Generals and the weather'' two-player game with uncertainty, formally defined in section \ref{sec:ExampleDef}.

\begin{example}[Extended equilibrium of the $\mathcal{GW}$ game]
\label{Ex:EEGW}
Direct calculation yields the following extended equilibrium $(\underline{\sigma}^*;\underline{\pi}^*)$ for the ``Generals and the weather'' ($\mathcal{GW}$) defined in section \ref{sec:ExampleDef} -- where $\underline{\sigma}^* = (\sigma_1^*, \sigma_2^*)$ and $\underline{\pi}^* = (\pi_1^*, \pi_2^*)$:

\begin{equation}
    \sigma_1^* = (p, \bar{p}), \quad
    \sigma_2^* = (q, \bar{q})
\end{equation}

\begin{equation}
    \pi_1^* = (P, \bar{P}), \quad
    \pi_2^* = (Q, \bar{Q})
\end{equation}

With the parameters:
\begin{table}[H]
    \centering
    \renewcommand{\arraystretch}{1.5}
    \begin{tabular}{|c|c|c||c|c|c|}
        \hline
        \multicolumn{3}{|c||}{\textbf{Parameters ($x$)}} & \multicolumn{3}{c|}{\textbf{Complements ($\bar{x} = 1-x$)}} \\
        \hline
        \textbf{Variable} & \textbf{Exact expr.} & \textbf{Num. approx.} & \textbf{Variable} & \textbf{Exact expr.} & \textbf{Num. approx.} \\
        \hline
        $p$ & $1-1/\sqrt{2}$ & 0.293 & $\bar{p}$ & $1/\sqrt{2}$ & 0.707 \\
        \hline
        $q$ & $2-\sqrt{2}$ & 0.586 & $\bar{q}$ & $\sqrt{2} - 1$ & 0.414 \\
        \hline
        $P$ & $1/\sqrt{2}$ & 0.707 & $\bar{P}$ & $1-1/\sqrt{2}$ & 0.293 \\
        \hline
        $Q$ & $\sqrt{2} - 1$ & 0.414 & $\bar{Q}$ & $2-\sqrt{2}$ & 0.586 \\
        \hline
    \end{tabular}
    \caption{Exact values and numerical approximations of the equilibrium parameters and their complements for the ``Generals and the weather'' ($\mathcal{GW}$) game's extended equilibrium.}
    \label{tab:variables}
\end{table}

To check that this extended strategy profile is indeed an extended equilibrium, we can directly calculate the expected utilities and expected regrets:

\begin{equation}
\begin{split}
& \mathrm{EU}_1(\text{``Up''} | \underline{\sigma}^*;\underline{\pi}^*) = P \ q + \bar{P} \ q = 2 - \sqrt{2} \\
& \mathrm{EU}_1(\text{``Down''} | \underline{\sigma}^*;\underline{\pi}^*) = P \ \bar{q} + \bar{P} \ (q+\bar{q}) = 2 - \sqrt{2}
\end{split}
\end{equation}

\begin{equation}
\begin{split}
& \mathrm{EU}_2(\text{``Up''} | \underline{\sigma}^*;\underline{\pi}^*) = Q \ \bar{p} = 1 - 1/\sqrt{2}\\
& \mathrm{EU}_2(\text{``Down''} | \underline{\sigma}^*;\underline{\pi}^*) = Q \ p + \bar{Q} \ p = 1 - 1/\sqrt{2}
\end{split}
\end{equation}

To derive the expected regrets, first, we can construct the action and state-dependent expected or effective utility matrices for both players, defined in \eqref{defeq:EUMatrix}  -- where rows represent the action of a given player, and columns the possible states of the weather:

\begin{equation}
    \underline{\underline{\mathrm{EU}}}^*_1
    :=
    \mathrm{EU}_1(.;.|\underline{\sigma}^*)
    =
    \begin{bmatrix}
        q & q \\
        \bar{q} & 1
    \end{bmatrix}, \quad
    \underline{\underline{\mathrm{EU}}}_2^*
    :=
    \mathrm{EU}_2(.;.|\underline{\sigma}^*)
    =
    \begin{bmatrix}
        \bar{p} & 0 \\
        p & p
    \end{bmatrix}
\end{equation}

By observing that $q > \bar{q}$ and $p < \bar{p}$ we arrive to the following effective regret matrices -- as defined in \eqref{eq:RegretMatrix}:

\begin{equation}
    \underline{\underline{\mathrm{ER}}}_1^*
    :=
    \mathrm{ER}_1(.;. | \underline{\sigma}^*)
    =
    \begin{bmatrix}
        0 & 1-q \\
        q-\bar{q} & 0
    \end{bmatrix}, \quad
    \underline{\underline{\mathrm{ER}}}_2^*
    :=
    \mathrm{ER}_2(.;. | \underline{\sigma}^*)
    =
    \begin{bmatrix}
        0 & p \\
        \bar{p} - p & 0
    \end{bmatrix}
\end{equation}

From which the expected regrets follow (yielding the same result that eq. \eqref{defeq:ER_from_Regret} would give):

\begin{equation}
\begin{split}
& \mathrm{ER}_1(\text{``Calm''} | \underline{\sigma}^*;\underline{\pi}^*) = (q - \bar{q}) \ \bar{p} = 3/\sqrt{2} - 2\\
& \mathrm{ER}_1(\text{``Storm''} | \underline{\sigma}^*;\underline{\pi}^*) = (1 - q) \ p = 3/\sqrt{2} - 2
\end{split}
\end{equation}

\begin{equation}
\begin{split}
& \mathrm{ER}_2(\text{``Calm''} | \underline{\sigma}^*;\underline{\pi}^*) = (\bar{p} - p) \ \bar{q} = 3 - 2 \sqrt{2}\\
& \mathrm{ER}_1(\text{``Storm''} | \underline{\sigma}^*;\underline{\pi}^*) = p \ q = 3 - 2 \sqrt{2}
\end{split}
\end{equation}

Therefore, because for all actions and states the expected utilities and regrets take the same values, the requirements for the extended equilibrium in Def. \autoref{def:ExEq} must hold (in fact, all inequalities in \eqref{defeq:ExtendedEquilibriumEU} and \eqref{defeq:ExtendedEquilibriumER} become equalities) with the following expected utilities and regrets:

\begin{equation}
    \begin{split}
        &\mathrm{EU}_1^* = 2 - \sqrt{2} \approx 0.586 \\
        &\mathrm{EU}_2^* = 1 - 1/\sqrt{2} \approx 0.293 \\
        &\mathrm{ER}_1^* = 3/\sqrt{2} - 2 \approx 0.121 \\
        &\mathrm{ER}_2^* = 3 - 2 \sqrt{2} \approx 0.172
    \end{split}
\end{equation}

\end{example}

\begin{remark}
    It can be shown by direct calculation that the extended equilibrium of the $\mathcal{GW}$ game is unique.
    By parametrizing all possible complete strategy profiles with the $p, q, P, Q \in [0,1]$ variables, one can reduce all $2+2+2+2=8$ inequalities defining the extended equilibrium to one single set of parameters\footnote{The calculation can be automated as well, symbolic mathematical software such as \textit{Mathematica 13.0} yields the result by using 
    \texttt{Solve[]} or \texttt{Reduce[]} \cite{reference.wolfram_2021_solve,reference.wolfram_2021_reduce}.} -- which equals to the presented equilibrium values in the previous example.
    
\end{remark}

\section{Existence of the Extended Equilibrium}

\begin{theorem}[Existence of the Extended Equilibrium]
\label{thm:ExExEq}
    For any finite game with uncertainty -- as defined in \autoref{def:GameWithUncertainty} -- there exists a complete strategy profile fulfilling the requirements -- defined in \autoref{def:ExEq} -- for the extended equilibrium.
    
\end{theorem}

To prove this general existence theorem, we first introduce a continuous map on the space of complete strategy profiles $\Upsilon : \mathcal{R} \mapsto \mathcal{R}$. To make the definition more concise, we introduce the Rectified Linear Unit, i.e. the $\mathrm{ReLU}(.)$ function \cite{paper:HintonReLU2010,paper:Fukushima1975,paper:Householder1941} in a standard way:

\begin{equation}
\label{eqdef:ReLU}
    \mathrm{ReLU}(x) = \max(0,x)
\end{equation}

\begin{definition}[Auxiliary map $\Upsilon$]
\label{def:Upsilon}
We define a continuous auxiliary map $\Upsilon : \mathcal{R} \mapsto \mathcal{R}$ by first introducing the functions $\varphi_{i,a_i}(.;.)$ and $\psi_{i,\theta}(.;.)$:

    \begin{equation}
    \label{eq:varphiDef}
        \varphi_{i,a_i}(\underline{\sigma};\underline{\pi}) = 
        \mathrm{ReLU} 
        \left (
        \mathrm{EU_i}(a_i|\underline{\sigma};\underline{\pi})
        -
        \mathrm{EU_i}(\underline{\sigma};\underline{\pi})
        \right )
    \end{equation}

    \begin{equation}
    \label{eq:psiDef}
        \psi_{i,\theta}(\underline{\sigma};\underline{\pi}) = 
        \mathrm{ReLU} 
        \left (
        \mathrm{ER}_i(\theta | \underline{\sigma};\underline{\pi})
        -
        \mathrm{ER}_i(\underline{\sigma};\underline{\pi})
        \right )
    \end{equation}

Then by introducing the mappings $\Phi(.;.)$ and $\Psi(.;.)$

    \begin{equation}
        \underline{\sigma}' = \Phi(\underline{\sigma};\underline{\pi})
        \iff
        \sigma_i'(a_i) = \frac{\sigma_i(a_i) + \varphi_{i,a_i}(\underline{\sigma};\underline{\pi})}{1 + \sum_{b_i \in \mathcal{A}_i} \varphi_{i,b_i}(\underline{\sigma};\underline{\pi})} 
    \end{equation}

    \begin{equation}
        \underline{\pi}' = \Psi(\underline{\sigma};\underline{\pi})
        \iff
        \pi'_i(\theta) = \frac{\pi_i(\theta) + \psi_{i,\theta}(\underline{\sigma};\underline{\pi})}{1 + \sum_{\chi \in \Theta} \psi_{i,\chi}(\underline{\sigma};\underline{\pi})} 
    \end{equation}

And finally, by combining these maps:
    
    \begin{equation}
        \Upsilon(\underline{\rho})
        =
        \Upsilon((\underline{\sigma};\underline{\pi}))
        =
        (\Phi(\underline{\sigma};\underline{\pi}); \Psi(\underline{\sigma};\underline{\pi})) 
        =
        (\underline{\sigma}';\underline{\pi}') 
        =
        \underline{\rho}'
    \end{equation}
    
\end{definition}

The simple but crucial observation is that the defined auxiliary map is a continuous self-map of complete strategy profiles. (All $\varphi_{i,a_i}$ and $\psi_{i,\theta}$ are non-negative, and the normalization in the definition of $\Phi$ and $\Psi$ ensures that all $\sigma_i'$ and $\pi_i'$ are proper probability distributions):

\begin{equation}
    \Upsilon \in C^0(\mathcal{R},\mathcal{R})
\end{equation}

The set of complete strategies $\mathcal{R}$ is a finite-dimensional compact and convex set. (Because it is the product of finite dimensional simplexes. The dimension of the space $\mathcal{R}$ can be given explicitly: $D = (|\Theta|-1)^N \ \prod_{i=1}^N (|\mathcal{A}_i|-1)$.)
Therefore, based on Brouwer fixed-point theorem \cite{misc:Jiang2007,book:FixedPointTheory,book:HandbookOfTopologicalFixedPointTheory,paper:NinetyYearsOfBrouwer,paper:Brouwer1911} there is a fixed point, providing the following simple lemma:

\begin{lemma}[Existence of a Fixed point]
\label{lemma:FixedPointExists}
The auxiliary map $\Upsilon(.)$ has (at least one) fixed point:
    \begin{equation}
    \exists \underline{\rho}^* \in \mathcal{R}, \quad \Upsilon(\underline{\rho}^*) = \underline{\rho}^* 
\end{equation}
\end{lemma}

The final step toward the proof of \autoref{thm:ExExEq} is the following equivalence theorem:

\begin{theorem}[Equivalence of Extended Equilibrium and Fixed points]
\label{thm:Equivalence}
A complete strategy profile is an extended equilibrium if and only if it is a fixed point of the $\Upsilon : \mathcal{R} \mapsto \mathcal{R}$ map. Formally:

    \begin{equation}
        \Upsilon(\underline{\rho}^*) = \underline{\rho}^* \iff
        \underline{\rho}^*=(\underline{\sigma}^*;\underline{\pi}^*) \text { is an extended equilibrium (as defined in \autoref{def:ExEq})}
    \end{equation}
\end{theorem}

\begin{proof}

We prove both directions of the statement.

{\bf Extended Equilibrium $\implies$ Fixed point:}
As we recall the definitions of $\varphi$ and $\psi$ functions in \autoref{def:Upsilon}, we see that because of the $\mathrm{ReLU}(.)$ function and the inequalities defining the extended equilibrium \eqref{eq:EU_ineq} and \eqref{eq:ER_ineq}, all $\varphi_{i,a_i}(\underline{\sigma}^*;\underline{\pi}^*)$ and $\psi_{i,\theta}(\underline{\sigma}^*;\underline{\pi}^*)$ functions has to be $0$ in an extended equilibrium.

Therefore, a complete strategy profile, which is an extended equilibrium, is also a fixed point of $\Upsilon(.)$.

\

{\bf Fixed point $\implies$ Extended equilibrium:}
For this second direction, we mainly follow the detailed steps in \cite{misc:Jiang2007}. The fixed point property of the complete strategy profile can be separated into a real and a fictional part:

\begin{equation}
    \Upsilon(\underline{\rho}^*) =
    \underline{\rho}^*
    =
    (\underline{\sigma}^* ; \underline{\pi}^*)
    \iff
    \underline{\sigma}^* = \Phi(\underline{\sigma}^* ; \underline{\pi}^*) \text{ and }
    \underline{\pi}^* = \Psi(\underline{\sigma}^* ; \underline{\pi}^*)
\end{equation}

Starting with the real strategies we get:

\begin{equation}
    \sigma_i^*(a_i) = 
    \frac{\sigma_i^*(a_i) + \varphi_{i,a_i}(\underline{\sigma}^*;\underline{\pi}^*)}{1 + \sum_{b_i \in \mathcal{A}_i} \varphi_{i,b_i}(\underline{\sigma}^*;\underline{\pi}^*)} 
\end{equation}

which implies the following equations:

\begin{equation}
\label{eq:Fixed_varphi}
    \varphi_{i,a_i}(\underline{\sigma}^*;\underline{\pi}^*)
    = \sigma_i^*(a_i) \ \lambda_i, \quad 
    \lambda_i = 
    \sum_{b_i \in \mathcal{A}_i} \varphi_{i,b_i}(\underline{\sigma}^*;\underline{\pi}^*)
\end{equation}

We will prove by contradiction that all $\lambda_i$ has to be $0$.
Let us assume that some $\lambda_i > 0$. In this case $\varphi_{i,a_i}(\underline{\sigma}^*;\underline{\pi}^*) > 0$ if $\sigma_i^*(a_i) > 0$ i.e. if $a_i \in \supp(\sigma_i^*)$. Because of the definition of the $\mathrm{ReLU}(.)$ function in \eqref{eqdef:ReLU} we have in this case:

\begin{equation}
\label{eq:varphiEU}
    \varphi_{i,a_i}(\underline{\sigma}^*;\underline{\pi}^*) > 0
    \implies
    \varphi_{i,a_i}(\underline{\sigma}^*;\underline{\pi}^*)
    =
    \mathrm{EU_i}(a_i|\underline{\sigma}^*;\underline{\pi}^*)
    -
    \mathrm{EU_i}(\underline{\sigma}^*;\underline{\pi}^*)
\end{equation}

Now let us take the following weighted sum of the set of equations in \eqref{eq:Fixed_varphi}:

\begin{equation}
    \sum_{a_i \in \supp(\sigma_i^*)} \varphi_{i,a_i}(\underline{\sigma}^*;\underline{\pi}^*)
    \ \sigma_i^*(a_i)
    =
    \lambda_i
    \sum_{a_i \in \supp(\sigma_i^*)} \left ( \sigma_i^*(a_i) \right )^2
\end{equation}

Substituting the expression for $\varphi_{i,a_i}(\underline{\sigma}^*;\underline{\pi}^*)$ from \eqref{eq:varphiEU} we get:

\begin{equation}
    \sum_{a_i \in \supp(\sigma_i^*)} 
    \left (
    \mathrm{EU_i}(a_i|\underline{\sigma}^*;\underline{\pi}^*)
    -
    \mathrm{EU_i}(\underline{\sigma}^*;\underline{\pi}^*)
    \right )
    \ 
    \sigma_i^*(a_i) 
    =
    \lambda_i
    \sum_{a_i \in \supp(\sigma_i^*)} \left ( \sigma_i^*(a_i) \right )^2
\end{equation}

which can be rearranged to the following form:

\begin{equation}
    \left (
    \sum_{a_i \in \supp(\sigma_i^*)} 
    \mathrm{EU_i}(a_i|\underline{\sigma}^*;\underline{\pi}^*)
    \ \sigma_i^*(a_i) 
    \right )
    -
    \mathrm{EU_i}(\underline{\sigma}^*;\underline{\pi}^*)
    \left (
    \sum_{a_i \in \supp(\sigma_i^*)} \sigma_i^*(a_i) 
    \right )
    =
    \lambda_i
    \sum_{a_i \in \supp(\sigma_i^*)} \left ( \sigma_i^*(a_i) \right )^2
\end{equation}

Now recalling the almost trivial Lemma \autoref{lemma:EUER}, we get:

\begin{equation}
    0
    =
    \lambda_i
    \sum_{a_i \in \supp(\sigma_i^*)} \left ( \sigma_i^*(a_i) \right )^2
\end{equation}

Therefore:

\begin{equation}
    \lambda_i = 0
\end{equation}

This contradicts with our assumption that some $\lambda_i > 0$, therefore in equilibrium:

\begin{equation}
    \forall i \ \lambda_i = \sum_{b_i \in \mathcal{A}_i} \varphi_{i,b_i}(\underline{\sigma}^*;\underline{\pi}^*) = 0
\end{equation}

and therefore by recalling equation \eqref{eq:Fixed_varphi}, in a fixed point:

\begin{equation}
    \forall i, \ \forall a_i \in \mathcal{A}_i, \ 
    \varphi_{i,a_i}(\underline{\sigma}^*;\underline{\pi}^*) = 0
\end{equation}

which by recalling equation \eqref{eq:varphiDef} and \eqref{eqdef:ReLU} implies that in a fixed point:

\begin{equation}
\label{eq:EU_ineq}
\boxed{
    \forall i, \ \forall a_i \in \mathcal{A}_i, \ 
    \mathrm{EU_i}(a_i|\underline{\sigma}^*;\underline{\pi}^*)
    -
    \mathrm{EU_i}(\underline{\sigma}^*;\underline{\pi}^*)
    \le 0
}
\end{equation}

We can repeat almost identical arguments and perform very similar steps to prove an analogous result for the fictional parts of the complete strategy profile (i.e. for the set of emergent subjective priors).

For the fictional strategies (or the emergent subjective priors) in a fixed point, we get:

\begin{equation}
    \pi^*_i(\theta) = \frac{\pi_i^*(\theta) + \psi_{i,\theta}(\underline{\sigma}^*;\underline{\pi}^*)}{1 + \sum_{\chi \in \Theta} \psi_{i,\chi}(\underline{\sigma}^*;\underline{\pi}^*)} 
\end{equation}

which implies the following equations:

\begin{equation}
\label{eq:Fixed_psi}
    \psi_{i,\theta}(\underline{\sigma}^*;\underline{\pi}^*)
    = \pi_i^*(\theta) \ \mu_i, \quad 
    \mu_i = 
    \sum_{\chi \in \Theta} \psi_{i,\chi}(\underline{\sigma}^*;\underline{\pi}^*)
\end{equation}

We will prove by contradiction that all $\mu_i$ has to be $0$.
Let us assume that some $\mu_i > 0$. In this case $\psi_{i,\theta}(\underline{\sigma}^*;\underline{\pi}^*) > 0$ if $\pi_i^*(\theta) > 0$ i.e. if $\theta \in \supp(\pi_i^*)$. Because of the definition of the $\mathrm{ReLU}(.)$ function in \eqref{eqdef:ReLU} we have in this case:

\begin{equation}
    \psi_{i,\theta}(\underline{\sigma}^*;\underline{\pi}^*) > 0
    \implies
    \psi_{i,\theta}(\underline{\sigma}^*;\underline{\pi}^*)
    =
    \mathrm{ER}_i(\theta | \underline{\sigma};\underline{\pi})
    -
    \mathrm{ER}_i(\underline{\sigma};\underline{\pi})
\end{equation}

Now let us take the following weighted sum of the set of equations in \eqref{eq:Fixed_psi}:

\begin{equation}
    \sum_{\theta \in \supp(\pi_i^*)} \psi_{i,\theta}(\underline{\sigma}^*;\underline{\pi}^*)
    \ \pi_i^*(\theta)
    =
    \mu_i
    \sum_{\theta \in \supp(\pi_i^*)} \left ( \pi_i^*(\theta) \right )^2
\end{equation}

Substituting the expression for $\psi_{i,\theta}(\underline{\sigma}^*;\underline{\pi}^*)$ from \eqref{eq:psiDef} we get:

\begin{equation}
    \sum_{\theta \in \supp(\pi_i^*)} 
    \left (
    \mathrm{ER}_i(\theta | \underline{\sigma}^*;\underline{\pi}^*)
    -
    \mathrm{ER}_i(\underline{\sigma}^*;\underline{\pi}^*)
    \right )
    \ \pi_i^*(\theta) 
    =
    \mu_i
    \sum_{\theta \in \supp(\pi_i^*)} \left ( \pi_i^*(\theta) \right )^2
\end{equation}

which can be rearranged to the following form:

\begin{equation}
    \left (
    \sum_{\theta \in \supp(\pi_i^*)} 
    \mathrm{ER_i}(\theta|\underline{\sigma}^*;\underline{\pi}^*)
    \ \pi_i^*(\theta) 
    \right )
    -
    \mathrm{ER_i}(\underline{\sigma}^*;\underline{\pi}^*)
    \left (
    \sum_{\theta \in \supp(\pi_i^*)} \pi_i^*(\theta) 
    \right )
    =
    \mu_i
    \sum_{\theta \in \supp(\pi_i^*)} \left ( \pi_i^*(\theta) \right )^2
\end{equation}

Now recalling the almost trivial Lemma \autoref{lemma:EUER} we get:

\begin{equation}
    0
    =
    \mu_i
    \sum_{\theta \in \supp(\pi_i^*)} \left ( \pi_i^*(\theta) \right )^2
\end{equation}

Therefore:

\begin{equation}
    \mu_i = 0
\end{equation}

This contradicts with our assumption that some $\mu_i > 0$, therefore in equilibrium:

\begin{equation}
    \forall i \ \mu_i = \sum_{\chi \in \Theta} \psi_{i,\chi}(\underline{\sigma}^*;\underline{\pi}^*) = 0
\end{equation}

and therefore by recalling equation \eqref{eq:Fixed_psi}, in a fixed point:

\begin{equation}
    \forall i, \ \forall \theta \in \Theta, \ 
    \psi_{i,\theta}(\underline{\sigma}^*;\underline{\pi}^*) = 0
\end{equation}

which by recalling equation \eqref{eq:psiDef}  and \eqref{eqdef:ReLU} implies that in a fixed point:

\begin{equation}
\label{eq:ER_ineq}
\boxed{
    \forall i, \ \forall \theta \in \Theta, \ 
    \mathrm{ER}_i(\theta | \underline{\sigma};\underline{\pi})
    -
    \mathrm{ER}_i(\underline{\sigma};\underline{\pi})
    \le 0
    }
\end{equation}

Equation \eqref{eq:EU_ineq} together with equation \eqref{eq:ER_ineq} prove that all fixed points are also extended equilibria.

\end{proof}

After this point, it is enough to gather the previously proved theorems and lemmas to prove \autoref{thm:ExExEq}:

\begin{proof}
{\bf Existence of Extended Equilibrium:}
Lemma \autoref{lemma:FixedPointExists} guarantees the existence of a fixed point of the auxiliary mapping $\Upsilon : \mathcal{R} \mapsto \mathcal{R}$;
while \autoref{thm:Equivalence} assures that a fixed point of $\Upsilon(.)$ is necessarily also a complete strategy profile in extended equilibrium.

This completes the proof of \autoref{thm:ExExEq}.
    
\end{proof}

\begin{remark}[Non-uniqueness]
    In general, the set of extended equilibrium solutions $\mathcal{R}^*$ does not need to contain only one complete strategy profile, i.e. games with uncertainty do not need to have a unique extended equilibrium solution.

    This is intuitive if we observe that if the uncertain parameters do not influence -- or only very weakly influence -- the players' utilities, then the extended equilibrium practically simplifies to a Nash equilibrium, therefore, inherits its non-uniqueness property (Coordination game and Battle of the Sexes being simple examples demonstrating the existence of multiple Nash equilibria) \cite{book:GameTheory}.
    
\end{remark}

\section{No emergent convictions}

An important property of the emergent subjective priors in an extended equilibrium solution is that -- if the uncertain parameter is not irrelevant -- no equilibrium prior will be concentrated on one state. In other words, no player will develop faith or conviction believing that the uncertain parameter has to have a specific value. 

\begin{theorem}[No Fictional Faith]
\label{thm:NoFictionalFaith}
    For any player $i$, a ``pure prior'', $\pi_i^{\eta}(\theta) = \delta_{\eta,\theta}$ (where $\delta$ stands for Kronecker delta\footnote{$\delta_{\eta,\theta} = 1$ if $\theta = \eta$ and $\delta_{\eta,\theta} = 0$ if $\theta \ne \eta$}) can be part of an extended equilibrium $(\underline{\sigma}^*, \underline{\pi}^*)$ only if the parameter is ``irrelevant'' for the $i$-th player, i.e. in the equilibrium there is an action (or actions) that is a best response -- there is no alternative action yielding higher expected utility -- regardless of the unknown parameter. 
    
    If a ``pure prior'' appears in an equilibrium solution for the $i$-th player, then  $\pi_i^*$ is totally degenerate \cite{book:AlgorithmicGameTheory,book:HandbookOfGameTheoryVol3}. Any prior $\pi'_i \in \Delta(\Theta)$ could be substituted to the subjective equilibrium prior of the $i$-th player $\pi_i^{\eta}$, preserving the conditions for the extended equilibrium. Therefore, a continuum set of extended equilibria solutions exists in this case.

\end{theorem}

The proof will be performed using the concepts of the expected or effective utility matrix, and effective regret matrix -- introduced in \eqref{defeq:EUMatrix} and \eqref{eq:RegretMatrix}.
Further important ingredients will be the simple properties of the effective regret matrix -- such as its non-negativity -- captured in Lemma \autoref{lemma:NonNegativeRegret} and Lemma \autoref{lemma:ExistenceBestResponse}.

\begin{proof}

    Let us assume that $\pi^\eta_i(\theta) = \delta_{\eta,\theta}$ is an emergent subjective prior for the $i$-th player in some extended equilibrium $(\underline{\sigma}^*,\underline{\pi}^*)$, i.e. $\pi_i^* = \pi^\eta_i$ for some $\eta \in \Theta$ value of the parameter.

    First, let us see what this assumption implies for the expected utilities expressed by the expected or effective utility matrix \eqref{defeq:EUMatrix}:

    \begin{equation}
        \mathrm{EU}_i^*(a_i) := \mathrm{EU}_i(a_i | \underline{\sigma}^*;\underline{\pi}^*) =
        \sum_{\theta \in \Theta} \mathrm{EU}_i(a_i;\theta | \underline{\sigma}^*) \ \pi_i^*(\theta) =
        \mathrm{EU}_i(a_i;\eta | \underline{\sigma}^*)
    \end{equation}

    \begin{equation}
        \mathrm{EU}_i^* := \mathrm{EU}_i(\underline{\sigma}^*;\underline{\pi}^*) =
        \sum_{a_i \in \mathcal{A}_i} \mathrm{EU}_i(a_i | \underline{\sigma}^*;\underline{\pi}^*) \ \sigma_i^* (a_i)
        =
        \sum_{a_i \in \mathcal{A}_i} \mathrm{EU}_i(a_i,\eta | \underline{\sigma}^*) \ \sigma_i^* (a_i)
    \end{equation}

    The first condition for the extended equilibrium \eqref{eq:EU_ineq} requires:

    \begin{equation}
        \forall a_i \in \mathcal{A}_i, \  \mathrm{EU}_i(\underline{\sigma}^*;\underline{\pi}^*) \ge \mathrm{EU}_i(a_i | \underline{\sigma}^*;\underline{\pi}^*)
    \end{equation}

    This means that for all $a_i \in \supp(\sigma_i^*)$ the expected utilities have to be equal: $\mathrm{EU}_i^*(a_i) = \mathrm{EU}_i^*$ -- the result can be viewed as a simple application of the indifference principle \cite{book:GameTheory}
    -- and for all $d_i \notin \supp{\sigma_i^*}$ they have to be less or equal: $\mathrm{EU}_i^*(d_i) \le \mathrm{EU}_i^*$

    From the definition of the effective regret matrix in \eqref{eq:RegretMatrix} and from the above-mentioned property of the expected utility, we get that the effective regret matrix has to be exactly zero for $(a_i,\eta)$ pairs, if $a_i \in \supp(\sigma_i^*)$:

    \begin{equation}
        \forall a_i \in \supp(\sigma_i^*), \quad \mathrm{ER}_i(a_i;\eta | \underline{\sigma}^*) = 0
    \end{equation}
    
    Therefore, the expected regret has to be zero:

    \begin{equation}
        \mathrm{ER}_i^* := \mathrm{ER}_i(\underline{\sigma}^*,\underline{\pi}^*) =
        \mathrm{ER}_i(\eta|\underline{\sigma}^*) =
        \sum_{a_i \in \mathcal{A}_i} \mathrm{ER}_i(a_i;\eta | \underline{\sigma}^*) \ \sigma_i^*(a_i) = 0
    \end{equation}

    Now we can use the second extended equilibrium condition from \eqref{eq:ER_ineq}:

    \begin{equation}
        \forall \theta \in \Theta, \  \mathrm{ER}_i(\underline{\sigma}^*;\underline{\pi}^*) \ge \mathrm{ER}_i(\theta | \underline{\sigma}^*;\underline{\pi}^*)
    \end{equation}

    Which by recalling \eqref{eq:ERthetaERMatrix} simplifies to:

    \begin{equation}
    \label{eq:RegretNonPositive}
        \forall \theta \in \Theta, \  \sum_{a_i \in \mathcal{A_i}} \mathrm{ER}_i(a_i;\theta | \underline{\sigma}^*) \ \sigma_i^*(a_i) \le 0
    \end{equation}

    Now we recall Lemma \autoref{lemma:NonNegativeRegret}, which states that all entries of the effective regret matrix have to be non-negative: 

    \begin{equation}
        \forall a_i \in \mathcal{A}_i, \ \forall \theta \in \Theta,  \quad \mathrm{ER}_i(a_i;\theta|\underline{\sigma}^*) \ge 0
    \end{equation}

    Lemma \autoref{lemma:NonNegativeRegret} and equation \eqref{eq:RegretNonPositive} together implies that:

    \begin{equation}
    \label{eq:RegretIsZero}
        \forall a_i \in \supp(\sigma_i^*), \ \forall \theta \in \Theta, \quad
        \mathrm{ER}_i(a_i;\theta|\underline{\sigma}^*) = 0
    \end{equation}

    This means that for all parameter values $\theta \in \Theta$, all actions $a_i \in \supp(\sigma_i^*) \subseteq \mathcal{A}_i$ are best responses, i.e. they give the maximal expected utility under any values of the parameter.

    From equation \eqref{eq:RegretIsZero} follows that for any probability distribution $\pi'_i \in \Delta(\Theta)$ the weighted sum of the regrets is $0$, satisfying the regret condition of the extended equilibrium \eqref{eq:ER_ineq}. 

    To see that any $\pi'_i \in \Delta(\Theta)$ satisfies the first extended equilibrium condition \eqref{eq:EU_ineq} as well, recall that -- because of \eqref{eq:RegretIsZero} and the definition of the regret \eqref{eq:RegretMatrix} -- we have:

    \begin{equation}
        \forall a_i \in \supp(\sigma_i^*), \ \forall \theta \in \Theta, \quad
        \mathrm{EU}_i(a_i;\theta|\underline{\sigma}^*) = 
        \max_{c_i} \left ( \mathrm{EU}_i(c_i;\theta|\underline{\sigma}^*) \right )
    \end{equation}

    and further

    \begin{equation}
        \forall d_i \notin \supp(\sigma_i^*), \ \forall \theta \in \Theta, \quad
        \mathrm{EU}_i(d_i;\theta|\underline{\sigma}^*) \le 
        \max_{c_i} \left ( \mathrm{EU}_i(c_i;\theta|\underline{\sigma}^*) \right )
    \end{equation}

    This means that for every state $\theta \in \Theta$ the following inequality holds:

    \begin{equation}
        \forall a_i \in \mathcal{A}_i, \ \forall \theta \in \Theta, \quad
        \sum_{b_i \in \mathcal{A}_i} \mathrm{EU}_i(b_i;\theta|\underline{\sigma}^*) \ \sigma_i^*(b_i)
        \ge 
        \mathrm{EU}_i(a_i;\theta|\underline{\sigma}^*)
    \end{equation}

    And therefore for any probability distribution $\pi'_i \in \Delta(\Theta)$ holds that:

    \begin{equation}
        \forall a_i \in \mathcal{A}_i, \quad
        \sum_{\theta \in \Theta} \sum_{b_i \in \mathcal{A}_i} \mathrm{EU}_i(b_i;\theta|\underline{\sigma}^*) \ \sigma_i^*(b_i) \pi_i'(\theta)
        \ge 
        \sum_{\theta \in \Theta}  \mathrm{EU}_i(a_i;\theta|\underline{\sigma}^*) \ \pi_i'(\theta)
    \end{equation}

    Or equivalently by using the notation for substitution:
    
    \begin{equation}
        \forall a_i \in \mathcal{A}_i, \quad
        \mathrm{EU}_i' = \mathrm{EU}_i(\underline{\sigma}^*;{\underline{\pi}^*}^{[i] \leftarrow \pi_i'})  \ge \mathrm{EU}_i(a_i | \underline{\sigma}^*;{\underline{\pi}^*}^{[i] \leftarrow \pi_i'})
    \end{equation}

    The expected utility and regret of the $j$-th player depends only on her subjective prior $\pi_j$ and the real strategy profile $\underline{\sigma}$, therefore the expected utility and regret for the $j$-th player does not change if we replace $\pi_i^*$ with $\pi_i'$ if $j \ne i$.

    Therefore, we can conclude that after changing $\pi_i^*$ to any $\pi'_i \in \Delta(\Theta)$, the full set of conditions for the extended equilibrium are satisfied:

    \begin{equation}
        \forall i, \forall a_i \in \mathcal{A}_i, \  \mathrm{EU}_i(\underline{\sigma}^*;{\underline{\pi}^*}^{[i] \leftarrow \pi_i'}) \ge \mathrm{EU}_i(a_i | \underline{\sigma}^*;{\underline{\pi}^*}^{[i] \leftarrow \pi_i'})
    \end{equation}
    \begin{equation}
        \forall i, \forall \theta \in \Theta, \  \mathrm{ER}_i(\underline{\sigma}^*;{\underline{\pi}^*}^{[i] \leftarrow \pi_i'}) \ge \mathrm{ER}_i(\theta | \underline{\sigma}^*;{\underline{\pi}^*}^{[i] \leftarrow \pi_i'})
    \end{equation}
    
\end{proof}

\begin{remark}
    \autoref{thm:NoFictionalFaith} -- the ``No Fictional Faith'' theorem -- ensures that the equilibrium priors satisfy a weakened regularity condition. In Bayesian inference,
    the so-called Cromwell's rule \cite{book:Lindley,book:BayesianSocialScience} -- also known as the regularity requirement \cite{paper:Regularity} -- suggests not to use a prior which associates $0$ or $1$ to specific events.
    In extended equilibrium -- if the parameter is not irrelevant to a player -- the subjective equilibrium prior will not associate probability $1$ to any specific parameter, therefore automatically satisfying the second half of Cromwell's rule --in cases where the parameter matters.

    (If the unknown parameter can take only two values, i.e. $|\Theta|=2$, the theorem guarantees the strong form of regularity.)
    
\end{remark}

\begin{example}[Irregular emergent subjective prior]
    To show that $0$ can indeed appear in a player's emergent subjective prior in extended equilibrium, consider the following simple single-player game with uncertainty:

    The player has two possible actions: $\mathcal{A_1} = \{\mathrm{U},\mathrm{D}\}$; while the uncertain parameter can be in three states: $\Theta=\{\mathrm{L},\mathrm{C},\mathrm{R}\}$.
    The utility function -- in this case representable by a single matrix -- is given below:

    \begin{equation}
        U_1(.;.) = 
        \begin{bmatrix}
            1 & 0 & 0 \\
            0 & 0 & 1
        \end{bmatrix}
    \end{equation}

    In this simple case, we can easily construct the effective regret matrix:

    \begin{equation}
        R_1(.;.) = 
        \begin{bmatrix}
            0 & 0 & 1 \\
            1 & 0 & 0
        \end{bmatrix}
    \end{equation}

    It is easy to see and verify that this game has one unique extended equilibrium, with the following strategy and emergent subjective prior pair:

    \begin{equation}
        \sigma_1^* = (1/2,1/2), \quad \pi_1^* = (1/2, 0, 1/2)
    \end{equation}

    \end{example}

    The previous simple example demonstrates that in some cases, zero probability might appear in the equilibrium prior.
    However, it is worth mentioning that the possibility of making an observation and collecting data leads to an extended game with a richer action set. In these extended games, new priors have to be obtained, which opens up the possibility to avoid the problem of ruling out potential states à priori, which might impact our actions.
    (The analysis and discussion of such issues will be performed in a separate paper, dedicated to statistical decision theory, i.e. single-player games with uncertainty.)
    
\section{Relation to other frameworks}

In this section, we show that the results for the ``Generals and the weather'' game (defined in \ref{ex:DefGeneralsAndTheWeather} and the results presented in \ref{Ex:EEGW})  cannot be simply obtained from other well-known constructions for games with incomplete information \cite{book:GameTheory}.

\paragraph{Games with chance nodes:}
Probably the simplest way to incorporate uncertain parameters into a game is to introduce a $N+1$-th ``Chance player'' -- often called ``Nature'' --, which is assumed to make a ``chance move'' with predefined probabilities $p_\theta$.
(We will not discuss this construction separately because its suggestions can be matched with those of the following Bayesian games.)

\paragraph{Bayesian games:}
Also known as Harsányi games\footnote{Called after John C. Harsányi, who introduced Games with Incomplete Information in \cite{paper:Harsanyi_I,paper:Harsanyi_II,paper:Harsanyi_III}.}, which introduces ``types'' for the players and a commonly known prior on these types \cite{paper:ZamirBayesianGames,book:GameTheory}. Formally a Bayesian game is a tuple: $\langle N, \underline{\mathcal{A}}, \underline{T}, \underline{U}, \pi_c \rangle$ where $U_i : \underline{\mathcal{A}} \times \underline{T} \to \mathbb{R}$ and $\pi_c \in \Delta(\underline{T})$.

\paragraph{Bayesian games with inconsistent beliefs:}
In this -- more general but less common -- framework, one needs to specify not only one common prior for the types but a set of priors or beliefs $\underline{\pi}$ \cite{book:GameTheory}.
One can define a -- generalized -- Bayesian equilibrium in games with incomplete information \cite{book:GameTheory}, which gives the same real equilibrium strategies $\underline{\sigma}^*$ of the extended equilibrium if the set of beliefs happens to be the equilibrium priors, i.e. $\underline{\pi} = \underline{\pi}^*$.

Therefore, the concept of extended equilibrium can be viewed as a prior selection mechanism for general Bayesian games with inconsistent beliefs.

\subsubsection{Bayesian Generals and the weather}

The ``Generals and the weather'' can be described both as a game with a chance node or as a Bayesian game. In the following we will attempt to embed the game into a Bayesian framework.

\begin{example}
As an attempt to embed the ``Generals and the weather'' to a Bayesian game, we introduce one extra player, ``Nature'', which can have two types: $T_N = \{ \text{``Calm''}, \text{``Storm''} \}$. Both generals can have only one single type $T_1 = \{*\}, T_2 = \{*\}$.
In this construction, we also need to specify $\pi_c \in \Delta(\underline{T})$, which, in this simple case, can be parametrized by one single probability $P_c \in [0,1]$ representing the probability of the weather -- or ``Nature'' -- being ``Calm''.
The first two players will have the same action sets as in \ref{ex:DefGeneralsAndTheWeather} $\underline{\mathcal{A}} = (\mathcal{A}_1,\mathcal{A}_2,\mathcal{A}_N)$, where $\mathcal{A}_1 = \{\text{``Up''},\text{``Down''}\}$, $\mathcal{A}_2 = \{\text{``Up''},\text{``Down''}\}$. 
However ``Nature'' will have only one trivial action: $\mathcal{A}_N = \{*\}$.

We will use the same utility functions as chosen in \ref{ex:DefGeneralsAndTheWeather} -- the only difference is they now depend on the type of ``Nature'' instead of the value of an uncertain parameter.

It is straightforward to calculate the Bayesian equilibrium strategies \cite{book:GameTheory} of this game as a function of $P_c$, yielding the result:

\begin{equation}
    \sigma^*_1(P_c) = \left(\frac{P_c}{1+P_c}, \frac{1}{1+P_c} \right ), \quad
    \sigma_2^*(P_c) = \left(\frac{1}{1+P_c}, \frac{P_c}{1+P_c} \right ), \quad
    \sigma^*_N=(1)
\end{equation}

(Generally, the strategies depend on the player's own type, but the Generals have no variation in type, and ``Nature'' has a trivial strategy independent of her type.)

\end{example}

\begin{remark}
    It is easy to check that there is no common prior parameter $P_c \in [0,1]$ which could yield equilibrium solutions equal to the quantities found as the extended equilibrium of the original problem in \ref{Ex:EEGW}.
    It is enough to observe that for any $P_c \in [0,1]$ a symmetry is present for the Bayesian equilibrium:
    \begin{equation}
        \sigma^*_1(P_c) (\text{``Up''}) = \sigma^*_2(P_c) (\text{``Down''})
    \end{equation}

    However, this would require an equality for the parameters found in the extended equilibrium solution in \ref{Ex:EEGW}, which is definitely not satisfied:

    \begin{equation}
        p = 1-1/\sqrt{2} \ne \bar{q} = \sqrt{2} - 1
    \end{equation}

    Another remarkable property of the extended equilibrium of the ``Generals and the weather'' game is that the -- partially subjective -- expected utilities of the players do not add up to $1$.

    \begin{equation}
        \mathrm{EU}_1^* + \mathrm{EU}_2^* = 2 - \sqrt{2} + 1 - 1/\sqrt{2}
        = 3 - \left ( \sqrt{2} + 1/\sqrt{2} \right ) 
        \approx 0.879
        < 1
    \end{equation}

    We have chosen the utilities to represent winning chances, and in this zero-sum (or, more accurately, constant-sum) game, they always add up to $1$. In all potential scenarios, either the ``Defender'' or the ``Attacker'' will win, and the other lose.

    The constant-sum property of the utility functions is preserved by the Bayesian equilibrium for any common prior $\pi_c$ but is violated by the -- partially subjective -- expected utilities in the extended equilibrium.
    The discrepancy can be interpreted as a ``perceived price of uncertainty''\footnote{Alternatively, the phenomena could be called the ``perceived price of ambiguity''.} -- relative to a setting where only risk or commonly known/believed randomness is present. 

    \end{remark}

    \begin{remark}
    Formally, one can calculate ``as if'' priors for both players. 
    If we observe only the strategy of the $i$-th player (the ``Defender'' or the ``Attacker'') -- presented in \ref{Ex:EEGW} --, we can deduce that they are behaving \emph{as if} a common prior were $P'_{c,i}$. The calculation of these ``as if'' priors yields the following result:

    \begin{equation}
        P'_{c,1} = \sqrt{2} - 1, \quad
        P'_{c,2} = 1/\sqrt{2}
    \end{equation}

    Curiously, when we compare these values with the results for the extended equilibrium priors in \ref{Ex:EEGW}, we see the quantities appear in reverse: $P'_{c,1} = Q$ and $P'_{c,2} = P$.
    
\end{remark}

\section{Conclusion}

In this paper, we have introduced a new framework that unifies non-cooperative game theory and statistical decision-making under a minimax regret perspective. By allowing each player to pick not only a mixed strategy but also an emergent subjective prior, we capture the essence of both strategic interaction and robust decision theory: players optimize their expected utility with respect to their own beliefs and simultaneously choose beliefs that hedge against regret in the face of global uncertainty. The existence of an extended equilibrium was established through a fixed-point argument that mirrors Nash's classical approach, demonstrating that, even without a shared prior, players can reach a coherent strategy-belief profile.

An important result is the ``No Fictional Faith'' theorem, which ensures that no player assigns full confidence to a single parameter value if that parameter matters to her potential utilities, thus providing a natural safeguard against completely dogmatic beliefs. Hopefully, this construction ``may appeal to some'' and may serve as a bridge between game-theoretic analyses and robust statistical decision approaches, helping decision-makers, scientists and agents of any sort striving to navigate real-world problems with deeply uncertain parameters.

\section{Acknowledgments}

I am deeply and sincerely grateful for all the support I got from my peers, friends and loved ones. Without the comments and feedback of Csongor Csehi, Dániel Soltész and Wicher Auguste Malten, this work would be more fragmented, partial and obscure.

Above all, I am deeply grateful to Anita L. Verő, who helped and supported me and was my trusted reader and my partner also in discussing the material.

\newpage

\bibliographystyle{plaindin}

\bibliography{ref}





\end{document}